\documentclass[USenglish,oneside,twocolumn]{article}
\usepackage[utf8]{inputenc}
\usepackage{amsfonts,amsthm,amsmath,amssymb}
\usepackage{hyperref}
\usepackage{framed}
\usepackage{array}
\usepackage{epsfig}
\usepackage{tikz}
\usepackage{enumitem} 
\usepackage{url}
\usetikzlibrary{positioning}

\newcommand{\ignore}[1]{}

\newtheorem{theorem}{Theorem}
\newtheorem{definition}[theorem]{Definition}

  \author{Saba Eskandarian, Eran Messeri, Joseph Bonneau, Dan Boneh}
  \title{\textbf{Certificate Transparency with Privacy}}
  
\begin{document}
\maketitle
  \begin{abstract}
{Certificate transparency (CT) is an elegant mechanism designed to
detect when a certificate authority (CA) has issued a certificate
incorrectly.  Many CAs now support CT and it is being actively
deployed in browsers.  However, a number of privacy-related challenges remain.
In this paper we propose
practical solutions to two issues. First, we develop a mechanism that
enables web browsers to audit a CT log without violating user privacy. Second, we 
extend CT to support non-public subdomains. 
}
\end{abstract} 

\section{Introduction}

There are many documented cases in which Certificate Authorities (CAs)
have issued certificates incorrectly. While DigiNotar and Comodo are
among the most well publicized examples~\cite{DigiNotar,Comodo},
misissuance happens several times a year~\cite{WoSign} and can enable
active man-in-the-middle (MITM) attacks on a large population of users.  For
example, misissuance of an {\tt example.com} certificate can lead to
an MITM on all {\tt example.com} traffic, unless defenses are deployed.

In response, several defenses have been proposed~\cite{CO13}, such as key
pinning~\cite{HPKP,TACK,eckersley2011sovereign,kim2013accountable,basin2014arpki}, DANE~\cite{DANE}, Perspectives~\cite{perspectives}, country-specific restrictions~\cite{kasten2013cage} and others.  Certificate Transparency (CT)~\cite{RFC} 
provides an elegant defense that is being actively deployed by web browser vendors and CAs.\\

\noindent \textbf{A brief overview of CT}\\ The goal of Certificate Transparency is to make all issued 
certificates publicly visible. By inspecting the set of issued certificates, domain owners can identify a certificate issued for 
their domain without permission and ensure that it is revoked.
To accomplish this, CT uses a set of public, untrusted,
append-only log servers that track and publish all
certificates issued by CAs.
Eventually, major browsers will only trust a certificate if it comes with a 
proof that it has been recorded in a public log. 
A certificate without such a proof will be treated as invalid.
This will effectively force all CAs to register every issued certificate 
with one or more CT logs.

Once CT is enforced by browsers, 
a rogue certificate for {\tt example.com} that is used in an attack 
must appear on one of the public logs.
This enables administrators at
{\tt example.com}, or an agent monitoring the logs on their behalf,
to detect the rogue certificate and revoke it. 
CT makes it possible to detect misissuance.
The task of investigating
and possibly penalizing negligent CAs is handled outside of the protocol.

To explain our work, we must first briefly review how CT works.  
CT adds an additional step to certificate issuance.  
When a CA wishes to issue a certificate for, say, {\tt example.com}, 
it chooses a public CT log and sends the certificate data, called a precertificate, to the log.  
The log uses a secret signing key to generate a
 {\em Signed Certificate Timestamp}, or SCT, which acts as a promise 
that the log will add the certificate to its public log within a specified 
period of time, called the {\em maximum merge delay} (MMD, usually 24 hours). 
\ignore{Figure~\ref{fig:sct} shows an example SCT.}
This SCT is sent back to the CA, and the CA typically embeds it in the final
signed certificate as an X.509 extension.

Two types of entities ensure that CAs and log servers properly follow CT procedures:
\begin{itemize}
\item {\bf Monitors} check for suspicious certificates in logs by downloading and reviewing all log entries.

\item {\bf Auditors} verify that logs are behaving correctly based on their partial views of logs and check that SCTs they encounter appear in logs.
\end{itemize}
Note that auditors can produce irrefutable cryptographic evidence of misbehavior by logs (such as attempting to delete log entries after inserting them), whereas monitors search for examples of misissued certificates which cannot be conclusively proven to indicate misbehaviour.
It is assumed that large organizations such as web hosting providers, content-delivery networks or CAs themselves will perform the role of monitors on behalf of their clients.
Large organizations such as Google or independent watchdogs will act as auditors to ensure logs are generally well-formed and updated properly.

In addition, web browsers can function as part-time auditors, periodically checking with logs for proof that 
SCTs for certificates from sites the browser has visited were indeed logged as promised.
Every log server stores its data in a Merkle tree, where
each leaf is a hash of a log entry.  This enables the log to
efficiently prove to the auditor that the tested certificate is recorded in
the log data using a standard Merkle proof-of-inclusion.  Auditors obtain an up to
date Merkle tree root, called a {\em Signed Tree Head (STH)}, via a
broadcast/gossip mechanism that we do not discuss here. \\

\ignore{
While CT is primarily designed for managing certificates, its
capabilities have found other applications.  For example, Mozilla
announced that all its software updates will be logged in CT, and the
Firefox browser will refuse to install an update, unless it can verify
that it has been logged on at least two CT log servers.  This ensures that one
cannot force Mozilla to covertly create a malicious software update
that will be deployed against a single targeted user.\\
}
\ignore{
\begin{figure}\centering
\fbox{
\parbox{5.9cm}{
\textbf{Version:} V1\\[2mm]
\textbf{LogID:} \\A4 B9 09 B4 18 58 14 87 BB 13 A2 CC 67 70 0A 3C 35 98 04 F9 1B DF B8 E3 77 CD 0E C8 0D DC 10 (Google US1 CT)\\[2mm]
\textbf{Timestamp:} \\Friday, November 1, 2016 8:37:28 PM \\[2mm]
\textbf{Extensions:} null \\[2mm]
\textbf{Signature} over all fields, including precertificate fields, but
  excluding LogID
}
}
\caption{Sample SCT embedded in a certificate.}
\label{fig:sct}
\end{figure}
}
\noindent \textbf{Privacy challenges with CT}\\
While CT provides a strong defense against misissuance, several privacy challenges are not addressed by the current
design and may hinder wide adoption.

First, CT auditing can violate users' browsing privacy.  Recall that when
an auditor, such as a web browser, encounters a valid SCT, it 
should check that the corresponding certificate is properly recorded on the
designated log server.  If the certificate is missing
from the log or, more precisely, if the log server fails to prove
inclusion, the browser must then publicize the SCT to indicate 
possible misbehavior by the log
server. A natural choice is to report the SCT 
to the browser's vendor (e.g., Google in the case of Chrome).
The vendor will need to investigate the log
and potentially remove the misbehaving log server from the
browser's list of trusted logs.

We note that this inclusion check by the auditor is primarily needed
for certificates that are not publicly accessible.  For public certificates, 
large auditors like Google
can check certificate inclusion in the log by themselves.

The problem with this approach is that it violates user privacy:
the browser must
send the offending SCT to the verifier, thereby revealing the user's
browsing behavior to the verifier.  This is especially troubling
considering that this auditing mechanism is primarily applied to
non-public sites that the user visits.  Ideally, we should enable log
server auditing without violating user privacy. 
Although, in principle, Tor could be 
used in this situation, we seek a solution that does not rely 
on external infrastructure. There are also situations where it is important to hide the certificate that has been excluded independent of whether the identity of the reporter is known.

\medskip
The second difficulty with CT is that it is currently incompatible with \textit{private} subdomains.
Consider an enterprise that does not want to reveal the domains of its
internal servers to the public.  However, the enterprise wishes to use
a public CA to issue certificates for its internal subdomains (or to log its privately issued certificates in a public log).
Because domain names are publicly available in the CT log, logging certificates will reveal the servers' private domains.

Certificates issued from private CAs (root CAs not trusted by default but manually added to browsers) are exempt from CT as a workaround to this privacy problem.
That is, browsers forgo the SCT requirement for certificate chains ending in a manually-installed root (a necessity to enable enterprise data loss prevention tools
to inspect HTTPS traffic).

However, this is unsatisfactory for two reasons.
First, installing a private root on all browsers is a major configuration burden.
More importantly, administering an internal CA introduces a significant security risk---if the internal CA is compromised then all client traffic may be eavesdropped.
Organizations may wisely wish to forgo this risk and rely on an external CA.

\ignore{
would be exempt from CT, but many dread the complexity
of running a CA and would rather use a public CA.  Alternatively, the
enterprise could ask the public CA to issue a certificate {\em
  without} an SCT and then manually install the certificate on all
employee machines.    However, by using an SCT-less certificate,
the enterprise is giving up on the benefits of CT.  This solution also adds
management complexity by requiring manual certificate installation on
every user device.  The question is whether there is a better solution
that simultaneously protects corporate privacy and permits certificate
transparency.
}

\subsection{Our Contributions}

We develop practical solutions to both challenges discussed above which can be
implemented efficiently on top of the existing CT
specification~\cite{RFC}.

\noindent \textbf{Privacy-preserving proofs of misbehavior}\\
We develop an efficient zero-knowledge protocol that enables an
auditor, such as a web browser, to prove to its vendor (say, Google)
that it has a valid SCT, properly signed by the log server, and yet
the log omits the corresponding certificate. 
Specifically, the auditor proves in zero-knowledge that it has a
valid SCT, as well as a valid proof of non-inclusion in the log's
Merkle tree.  This proves log misbehvaior.
The vendor learns nothing, other than the fact that log
integrity has failed.  

The zero-knowledge proof is about 330\;KB and takes approximately five
seconds to generate and three seconds to verify.  This overhead
is tolerable, given that this mechanism is used infrequently.  We also
consider a variant that provides weaker privacy guarantees, but
effectively revokes the missing certificate without revealing it to
the verifier. Our
construction is presented in Section~\ref{sec:ZK}, where we discuss
several practical considerations in its application.

\medskip \noindent \textbf{CT for private subdomains}\\
We give a complete solution that makes CT fully compatible with private
subdomains.
Our construction makes it possible for enterprises to use a public CA 
that issues standard SCTs without revealing any information about the names of internal domains.
Our construction, presented in Section~\ref{sec:privateCT}, 
uses only commitments and has very low overhead.

\medskip \noindent \textbf{Short-lived certificates}\\
We also consider the implications of using CT in conjunction with
short-lived certificates, which are certificates that are valid for
only one day~\cite{TSH+12}.  The two ideas, CT and short-lived
certificates, seem incompatible: the large number of short-lived
certificates would overwhelm the log servers.  In
Section~\ref{sec:shortLived} we provide a simple solution, using Merkle trees, that resolves this tension.

\section{Zero Knowledge Proof of Exclusion}
\label{sec:ZK}

We start by showing how to audit CT
logs without violating user privacy.  
Specifically, we show how a web browser can
construct an efficient zero-knowledge proof that a log server has
issued a valid SCT but that the corresponding certificate has not been
entered into its Certificate Transparency (CT) log.  This proves that
log integrity has failed without revealing any information about
the user's browsing behavior.  

Throughout the section we use the following terminology:
\begin{itemize}
\item {\bf Log:} The entity managing a CT log and issuing SCTs that
  must be inserted into the log.

\item {\bf Prover:} An auditor, such as a web browser, that obtained
  a certificate and a SCT by visiting a site, but the certificate
  data is missing from the designated log.

\item {\bf Verifier:} An authority that wants to learn that a log is
  misbehaving without learning the prover's browsing behavior.
  That is, without learning the offending SCT. This authority has the
  power to investigate the log to determine what went wrong, and
  to potentially take steps to revoke the log.
\end{itemize}

More precisely, a log entry is a tuple $(\textit{data}, I, T)$, where
$I$ and $T$ are 64-bit integers representing the index of the entry
among the leaves of the Merkle tree and its timestamp respectively.
$\textit{data}$ represents the rest of the contents of the log
entry (e.g. an X509 certificate). Recall that log entries make up the
leaves of a Merkle tree, sorted by their timestamp~$T$. An SCT is the
pair $(\textit{data}, T)$, where $\textit{data}$ contains the log's
signature on the domain name and other certificate data.  When $x$
represents a log entry or an SCT, we use $T_x$ to refer to the
timestamp of that entry and $I_x$ to refer to that entry's index in 
the leaves of the log's Merkle tree.

In our system, the prover proves that it holds an SCT~$y$ whose
timestamp $T_y$ falls between the timestamps of two adjacent log
entries~$x$ and~$z$ in the Merkle tree. This proves that the entry~$y$
is missing from the log because it would otherwise appear between~$x$
and~$z$. The proof should reveal nothing to the verifier beyond the
fact that some SCT is missing from the log.

\subsection{Privacy Goals and Limitations}

We begin by discussing the threat model, and the level of
privacy to be expected from a solution. 

Recall that the prover wishes to prove to a verifier that
it holds an SCT $y$ whose timestamp falls between two neighboring log
entries $x$ and $z$ without revealing $x$, $y$, or $z$. While it is
clear why revealing the missing SCT $y$ would reveal information about
the prover's web browsing, it may not immediately be obvious why $x$
and $z$ should also be hidden. The reason is that knowledge of the
timestamps of $x$ and $z$ would allow the verifier to learn what users
have visited some particular suspected site. The verifier simply visits that
site and checks if its SCT timestamp does indeed fall between the
$x$ and $z$ reported by the prover. We emphasize that it is important
not only to prevent leaking information about the prover's browsing
habits but also to avoid giving the verifier information it could use
on its own to uncover the sites visited by the prover.

Suppose that a log only records SCTs for publicly accessible domains,
and that only a single SCT is missing from the log, say belonging to
domain $D$. Once the prover sends the verifier its zero-knowledge
proof that the log is missing some SCT, the verifier could do an
exhaustive search of all public SCTs and learn that the SCT for domain $D$
has been excluded. This reveals to the verifier that the prover
visited domain $D$. This privacy limitation is inherent to CT.

Our approach is primarily used to preserve privacy when a prover (web
browser) audits a private domain that is not publicly accessible.  It
also provides privacy in case the log drops many SCTs from its log
data.

\medskip
An alternative way to protect the privacy of the prover's web behavior
is to have the prover submit its accusation to the verifier over
Tor. Our solution avoids relying on external infrastructure like Tor,
and does not rely on non-collusion assumptions that Tor requires. Moreover, there are situations where it is important not only to hide the identity of the reporter but also the identity of the site being visited. Consider for example if the excluded SCT corresponds to a private domain. Sending the missing certificate to a verifier, even if it could be done without revealing the identity of the reporter, still reveals the private domain to the verifier. This means that private reporting of exclusions is important not only to protect the web history of a user but also to protect private domains on CT logs. 

\medskip
Recall that there is a period of time, the {\em maximum merge delay}
(MMD), after an SCT is issued and before it is required to be present
in the log.  One concern is that this may cause ``false positives''
where a certificate with a valid SCT can be shown to be excluded from
the log before the MMD period has elapsed. As we will see, this 
does not affect our solution because the prover must present a log entry
that has been added after the MMD has elapsed. 

\medskip
\noindent \textbf{Actionable evidence}\\
Our zero-knowledge proof hides, by design, all information about the
log's infraction, except the fact that some SCT is missing.  This may
make a subsequent investigation of the log server more
challenging. Once a log has been shown to have excluded SCTs, it must
launch an internal investigation to determine the source of the
problem and satisfy investigators that its practices are adequate. At
the very least, the log can conclude that it has been compromised and make
sure to change its signing keys for future operation.

In Section~\ref{sec:actionable} we present an alternate approach that,
in addition to proving log-exclusion, lets other browsers treat
the missing certificate as invalid.  
This effectively revokes the offending certificate while revealing
nothing to the verifier.

\subsection{Preliminaries}
Our construction requires the following additional components in the CT protocol:\\

\noindent \textbf{Additional signatures:} Each log entry $x$ is accompanied by signed messages $$\textit{Sign}_{k_H}(H(x)),\ \ 
  \textit{Sign}_{k_T}(T_x+H(x)),\ \ 
  \textit{Sign}_{k_I}(I_x+H(x)),$$ 
where $k_H, k_T,$ and $k_I$ are different signing keys. Each SCT $y$ is accompanied only by one signed message $\textit{Sign}_{k_T}(T_x+H(x))$. In total, this requires 4 additional signed messages to be distributed by the log for each certificate: 1 signature in each SCT and 3 signatures in each log entry.\\

\noindent \textbf{Ordered log entries:} Monitors must make sure that the ordering of index numbers corresponds to the order of timestamps in log entries. That is, logs must ensure that if $I_x>I_y$, then it must hold that $T_x>T_y$. This ensures that the logs fulfill our definition of being well-formed and contrasts with the current setting where indexes can be assigned to queued entries in any order in a sequencing phase after SCTs are distributed.\\

\noindent
We also need the following specialized primitives:\\

\noindent \textbf{Commitments:} We need a commitment scheme that is additively homomorphic and supports efficient zero-knowledge equality tests and range proofs. We informally define a commitment with binding and hiding properties as follows:
\begin{itemize}
\item[--] \textbf{Binding:} Given a commitment $C_m$ to message $m$, it is computationally hard to decommit it to any message $m'\neq m$.
\item[--] \textbf{Hiding:} It is computationally hard to determine the message $m$ given only commitment $C_m$.
\end{itemize}

An additively homomorphic commitment scheme is one which allows for addition of two values while both are under commitments. Equality tests check whether the values of two commitments are equal, and range proofs prove that the value of a commitment falls within some range (or in our case checks that the value is greater than zero). Many practical schemes and protocols with these properties exist \cite{Ped92, DF01, CM01, Bou00, DBB+15, CP93, Sch91, Bra97, CEvdG87, CCS08}.\\

\noindent \textbf{Signatures:} We also require a signature scheme with efficient proofs of 
knowledge of a signature:
\begin{definition}[Proof of Knowledge of a Signature]  \label{def:sigcom}
A Proof of Knowledge of a signature is a zero knowledge proof that
$(pk', c_1', c_2')\in P$, where $P$ is defined as the language of
triples $(pk, c_1, c_2)$ where $c_1$ and $c_2$ are commitments to $m$
and $\textit{Sign}_{pk}(m)$, respectively, for some message $m$. 
Moreover, there exists an extractor $E$ that, given black box access to the prover, can extract the message $m$ and signature $\textit{Sign}_{pk}(m)$.
\end{definition}
The following signature schemes have efficient proofs of knowledge of a signature:
\begin{itemize}
\item CL02 \cite{CL02} signatures, based on RSA
\item BBS \cite{BBS04} signatures, based on pairings, can be easily modified, as described Camenisch and Lysyanskaya \cite{CL04}, to provide efficient proofs of knowledge of signatures
\item CL04 \cite{CL04} signatures, also based on pairings
\end{itemize}

\noindent \textbf{Hash function:} Finally, we will use a near collision resistant hash function, defined as follows:

\begin{definition}[$\delta$-near collision resistant hash function]
A hash function $H: \{0,1\}^*\rightarrow Z_n$ is $\delta$-near collision resistant if it is hard to find $x,y\in\{0,1\}^*$ such that $|H(x)-H(y)|<\delta$.
\end{definition}

Specifically we will assume $H: \{0,1\}^*\rightarrow Z_{2^\lambda}$ is a $\textit{poly}(\lambda)$-near collision resistant hash function, where $\lambda$ is a security parameter and $\textit{poly}(\lambda)$ denotes some polynomial function of $\lambda$.

Although a stronger assumption than collision-resistance, it appears that most hash functions conjectured to be collision-resistant also satisfy $\textit{poly}(n)$-near collision resistance. If $H$ is modeled as a random oracle, then $H$ has $\textit{poly}(n)$-collision resistance since the probability that two random points in $\{0,1\}^n$ fall within $\textit{poly}(n)$ of each other is $\frac{\textit{poly}(n)}{2^n}$.

\subsection{Construction}
Protocol $\Pi$ is between a prover and a verifier with capabilities to somehow punish a misbehaving log. The zero knowledge proof is in two parts. In the first part (summarized in figure \ref{picture}), the verifier receives commitments $C_{T_x}$, $C_{I_x}$, $C_{T_y}$, $C_{T_z}$, and $C_{I_z}$ to the indexes $I$ and timestamps $T$ of log entries $x, z$ and to the timestamp of the SCT $y$. The verifier receives additional commitments described below, including a commitment to $H(x)$. The second part proves that the indexes $I_x$ and $I_z$ are adjacent and that the timestamps are ordered such that $T_x<T_y<T_z$.

The first part uses the signed messages accompanying a log entry to both verify that the commitments given by the prover are legitimate and to prove that each $(\text{index}, \text{timestamp})$ pair corresponds to the same log entry. Using the additive homomorphic properties of the commitment, both $I_x$ and $T_x$ are added to $H(x)$, giving commitments to $I_x+H(x)$ and $T_x+H(x)$. The prover then proves in zero knowledge that the values in these commitments are properly signed by the log. 

To complete the actual proof of omission, the prover sends and reveals a commitment $C_1$ to the number $1$ and proves in zero knowledge that $C_{I_x+1}$ and $C_{I_z}$ are commitments to the same value. Next, the verifier computes commitments $C_{p_1}=C_{T_z-T_y}$ and $C_{p_2}=C_{T_y-T_x}$. Finally the prover proves in zero knowledge that $p_1>0$ and $p_2>0$. 

The complete protocol $\Pi$ is as follows:

\begin{enumerate}
\item Prover sends commitments $C_{\textit{Sign}_{k_I}(I_x+H(x))}$, $C_{\textit{Sign}_{k_H}(H(x))}$, $C_{\textit{Sign}_{k_T}(T_x+H(x))}$, $C_{\textit{Sign}_{k_T}(T_y+H(y))}$, $C_{\textit{Sign}_{k_I}(I_z+H(z))}$, $C_{\textit{Sign}_{k_H}(H(z))}$, $C_{\textit{Sign}_{k_T}(T_z+H(z))}$, $C_{T_x}$, $C_{H(x)}$, $C_{I_x}$, $C_{H(y)}$, $C_{T_y}$, $C_{T_z}$, $C_{H(z)}$, $C_{I_z}$, $C_1$, and $r_{C_1}$.

\item Verifier computes $C_{I_x+H(x)}$, $C_{T_x+H(x)}$, $C_{T_y+H(y)}$, $C_{I_z+H(z)}$, and $C_{T_z+H(z)}$

\item Prover proves in zero knowledge that each commitment to a signature in step~1 corresponds to the correct value, as shown in figure \ref{picture}. That is, the prover verifies a signature under a commitment for each edge labeled ``verify'' in figure \ref{picture}.  Here we are using three zero-knowledege proofs for the language $(pk', c_1', c_2')$ where $c_2'$ is a commitment to a valid signature under $pk'$ for a message committed to by $c_1'$.

\item Verifier computes $C_{I_x+1}$, $C_{p_1}=C_{T_z-T_y}$, and $C_{p_2}=C_{T_y-T_x}$.

\item Prover proves in zero knowledge that $I_x+1=I_z$, $p_1>0$, and $p_2>0$.
\end{enumerate}

\noindent \textbf{Complexity} \\
The overhead of our scheme does not increase as the size of the log grows, meaning we achieve $O(1)$ proof construction and verification time in the size of the log. This optimal asymptotic overhead is enabled by the strict ordering of log entries by timestamp enforced by monitors.  Imposing such a locally checkable condition on the structure of the log removes the need for any computation that would involve more than one SCT and two log entries in either the construction of a proof or its verification. Our construction also enjoys reasonable overhead in practice, as discussed in section \ref{PerfEval}.

\begin{figure}
\centering
\begin{tikzpicture}
\node (sighash) {$C_{\textit{Sign}_{k_H}(H(x))}$};
\node (hash) [below=of sighash] {$C_{H(x)}$} ;
\node (Thash) [right=of hash]{$C_{T_x+H(x)}$};
\node [style=rectangle, draw=black, very thick] (T) [below=of Thash] {$C_{T_x}$};
\node (sigT) [above =of Thash] {$C_{\textit{Sign}_{k_T}(T_x+H(x))}$};

\node (Ihash) [left=of hash]{$C_{I_x+H(x)}$};
\node [style=rectangle, draw=black, very thick] (I) [below=of Ihash] {$C_{I_x}$};
\node (sigI) [above =of Ihash] {$C_{\textit{Sign}_{k_I}(I_x+H(x))}$};

\draw[->] (hash.east) to node[fill=white] {\tiny sum} (Thash.west);
\draw[->] (hash.north) to node[fill=white] {\tiny verify} (sighash.south);
\draw[->] (Thash.north) to node[fill=white] {\tiny verify} (sigT.south);
\draw[->] (T.north) to node[fill=white] {\tiny sum} (Thash.south);
\draw[->] (sigT.south) to node[fill=white] {\tiny verify} (Thash.north);

\draw[->] (hash.west) to node[fill=white] {\tiny sum} (Ihash.east);
\draw[->] (sighash.south) to node[fill=white] {\tiny verify} (hash.north);
\draw[->] (sigI.south) to node[fill=white] {\tiny verify} (Ihash.north);
\draw[->] (I.north) to node[fill=white] {\tiny sum} (Ihash.south);
\draw[->] (Ihash.north) to node[fill=white] {\tiny verify} (sigI.south);

\end{tikzpicture}

\caption{The process of getting verified commitments to the Index and Timestamp of a log entry. SCTs only require one signature because there is no need to ensure a connection with an Index for SCTs (which have no Index).}
\label{picture}
\end{figure}
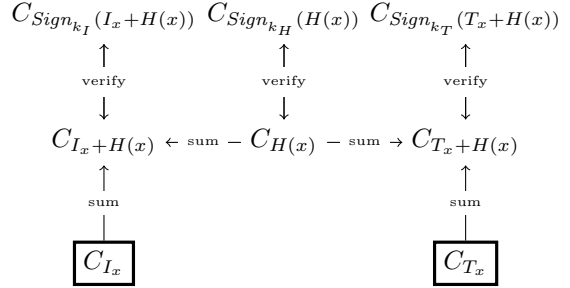

\subsection{Security Analysis}

We now prove the security of the construction in the previous section. Appendix \ref{ZKApp} gives a formal definition of a CT log with the modifications we require for our scheme to work, a security definition that captures our goals, and a full proof of security. Informally, we require the following properties of a Zero-Knowledge Exclusion Proof:
\begin{itemize}
\item \textbf{Completeness}: Any prover who really does have an SCT that has been excluded from a log will be able to convince a verifier of this fact.

\item \textbf{Soundness}: No fraudulent prover can convince a verifier that a log has excluded an SCT. We define this in terms of a security game $\textit{ProofExcl}$ between a verifier $V$ and an adversary $A$ who wishes to run protocol $\Pi$ to prove that honest log $L$ has excluded an entry. We require that no adversary $A$ can win this game (and therefore convince $V$ to accept a fraudulent proof) with more than negligible probability. 

\item \textbf{Zero-Knowledge}: A verifier learns nothing from the proof except that a log entry has been excluded.  We achieve this by exhibiting a simulator for the verifier's view of the protocol. 
\end{itemize}

The proof of security is given in Appendix~\ref{ZKApp}.
Here we sketch the main idea.
The proofs for completeness and zero knowledge follow from the properties of the underlying components of the protocol. Proving soundess requires more work. The idea is that the adversary is forced by the security of the signatures to use only index and timestamp values that appear in the log. We need to show that the log cannot contain signed sums $s_1 = I_x+H(x)$ and $s_2 = I_z+H(z)$, and signed hashes $H(x')$ and $H(z')$ such that $s_1 - H(x')$ is one less than $s_2 - H(z')$.  Otherwise, the adversary would obtain a fake pair of index numbers that are one apart and that can therefore fool the verifier (a similar statement should also hold for timestamps). Finding such entries would, however, break the $\textit{poly}(\lambda)$-near collision resistance of the hash function, as shown in Appendix~\ref{ZKApp}.

\subsection{Alternative Construction}

The zero-knowledge proof in the previous section can be modified so that it only relies collision resistant hash functions, not $\delta$-near collision resistance.  However, this proof is slightly less efficient. 

The modified proof $\Pi'$ is similar to $\Pi$, except that instead of having signatures on $I_x+H'(x)$ and $T_x+H'(x)$, the log publishes signatures on $I_x \| H'(x)$ and $T_x \| H'(x)$, where $\|$ denotes concatenation. Before adding $T_x$ or $I_x$ to $H'(x)$, the verifier multiplies each by a commitment to $2^n$, effectively left-shifting $T_x$ and $I_x$ so the sum with $H'(x)$ can be interpreted as a concatenation of the two values. This simplifies the soundness proof so that the value of $d$ will always be 0, and any adversary $A'$ winning the modified security game could be used by another adversary $B'$ to find a collision for $H'$. The proof that $\Pi'$ is a zero-knowledge proof is otherwise very similar to that of $\Pi$. Although the soundness of $\Pi'$ relies on a weaker assumption than $\Pi$, it requires the prover to send an additional commitment to the verifier and for the prover and verifier to compute two interactive multiplications on commitments \cite{CM01}, rendering $\Pi'$ less efficient than $\Pi$. 

\subsection{Actionable Proof of Exclusion}
\label{sec:actionable}

Although we presented a solution with best-possible privacy, in some contexts it is desirable to reveal some more information to help with deployment of the system. In particular, we are interested in the case where the verifier needs some recourse other than completely distrusting a log after learning that the log has cheated. In our current system, a verifier has no power, for example, to issue a whitelist or blacklist of trusted/untrusted certificates for a particular log because it learns nothing about excluded SCTs other than that they exist.

We can do much better if we are willing to compromise on the degree of privacy offered to provers.  Since it is likely that any SCT excluded from the log but publicly accessible on the internet will be caught by an auditor (e.g. Google) without need for a proof, we must mainly concern ourselves with SCTs for certificates that are not publicly accessible.  We want a scheme where the verifier does not learn the certificate corresponding to a bad SCT but enables browsers to reject that certificate if they encounter it.  In other words, we wish to enable browsers to recognize when an SCT they encounter has been excluded from a log without revealing the contents of the SCT to a verifier.  

We achieve this by relaxing our privacy goals. Whereas our original proof was of a statement of the form ``There exists a properly signed SCT $y$ whose timestamp falls between adjacent log entries $x$ and $z$,'' we now wish to prove that ``There exists a properly signed SCT $y$ whose timestamp falls between adjacent log entries $x$ and $z$, and $c = H(y)$ is the hash of $y$,''
where $c$ is supplied by the prover. By constructing a proof that intentionally reveals $H(y)$, we make it possible for an auditor who encounters an SCT $y'$ to check if $H(y')=H(y)$. At the same time, the high entropy of $H(y)$ ensures that $y$ remains hidden. This weakened proof is of course only useful if the SCT $y$ in question is not in the set of publicly available SCTs or else the verifier could test $H(y)$ against hashes of each available SCT and discover $y$ by exhaustive search. Fortunately, the scenario where a non-public SCT is omitted from the log is exactly the scenario we aim to address.

We need only to make a slight modification to our zero knowledge protocol in order to achieve this weakened notion of privacy. Suppose we choose to instantiate our protocol with Pedersen commitments \cite{Ped92}, which provide all the properties we require in a commitment scheme. Pedersen commitments are of the form $C_m=g^mh^r$ for public $g, h$ and secret randomness $r$. We can set $r=0$ in our commitment to $H(y)$ and maintain the binding property of the commitment due to the high entropy of $H(y)$, as mentioned above (and have the verifier compute the commitment $g^{H(y)}$ upon receiving $H(y)$). The rest of the proof proceeds as before. Now the verifier can distribute $H(y)$, and any browser can check if an SCT it has acquired matches the one used to construct the proof.

\subsection{Practical Considerations} \label{ZKpractical}

The presented solution can be implemented on top of the existing Certificate Transparency specification, RFC6962, with some adjustments.\\

\noindent \textbf{Key distribution:} The existing mechanism for distributing the log's key for verifying SCTs signed using $k_S$ can be used for distributing the additional keys for verifying signatures produced using $k_H$, $k_T$, and $k_I$.\\

\noindent \textbf{Obtaining entries around a timestamp:} The current CT RFC only specifies an API for getting the index of an entry in the log given its hash, requiring the client to know the timestamp and the certificate for that entry. A new API endpoint is needed to provide a number of log entries around a timestamp specified by the client:\texttt{/ct/v1/get-entries-around-timestamp} with the following inputs:
\begin{itemize}[nosep]
  \item {\em timestamp:} The desired timestamp, in decimal.
  \item {\em count:} The desired number of entries.
\end{itemize}

\medskip\noindent
and the following output for each entry:
\begin{itemize}[nosep]
  \item {\em index:} Index of the entry.
  \item {\em sct:} JSON structure containing the SCT, as described in Section 4.1 of RFC6962.
  \item {\em signatures:} JSON structure containing $\textit{Sign}_{k_H}(H(e))$, $\textit{Sign}_{k_I}(i+H(e))$ and $\textit{Sign}_{k_T}(t+H(e))$ ($H(e), i+H(e)$ and $t+H(e)$ can be calculated by the client).
\end{itemize}
This additional API endpoint may be independently useful for monitors searching for fraudulent certificates suspected to be issued around the time of some particular event, e.g. a security breach. \\

\noindent \textbf{Adding signatures to SCTs:} Signed Certificate Timestamps have an extensions field which can be used for embedding the signatures $\sigma_{T_y+H(y)},$ $\sigma_{H(y)}$. As the signature in the Signed Certificate Timestamp covers the extensions field, these signatures should be produced before the signature over the entire Signed Certificate Timestamp, $\textit{Sign}_{k_S}(s)$.\\

\noindent \textbf{Uniqueness of timestamps:} Timestamps are required to be unique under this construction. Distributed implementation of a log could issue two SCTs with exactly the same timestamp for two different entries submitted at the same time. To avoid restricting the implementation, each front-end of the log will be assigned a unique ID, which will be concatenated to the timestamp to ensure uniqueness. In practice, the unique ID should be added as an extension to the SCT to avoid changing the timestamp format currently specified in the CT RFC.\\

\noindent \textbf{Denial of Service:} The need for the verifier to check several signatures and zero knowledge proofs introduces the potential for a denial of service attack launched by sending many fake proofs that will fail to verify but will occupy verifier time, thereby preventing proper functioning of this kind of proof of exclusion. Although concern regarding denial of service potential is legitimate, it is not as bad as it may first seem. The verifier can abort verification of a bad proof early in less than 1/7 of the verification time for a correct proof when either a signature verification fails or the submitted commitments do not satisfy the relations needed for the proof to go through.\\

\noindent \textbf{Uncooperative logs:} Logs have an incentive to not respond to queries by auditors who are building a proof of exclusion. Suppose a log does not respond to timestamp-based queries.  An auditor can instead use index-based queries (also used by monitors) and perform a binary search to find the entry with the relevant timestamp needed for constructing the proof. 

A malicious log could go further.  Suppose it  places a ``dummy'' entry on each side of the missing entry and refuses to respond to requests for those dummy entries. This prevents the auditor from constructing the proof.  However, this tactic will be unsuccessful because monitors will discover this misbehavior when they attempt retrieve all log entries as part of their normal periodic sweep.  For this we need that a request to a CT log from a monitor is indistinguishable from a request from an auditor. This is needed to defeat uncooperative logs, and is also needed for correct operation of log monitoring. 

\subsection{Performance Evaluation} \label{PerfEval}
\begin{figure*}
\begin{center}
    \small
    \begin{tabular}{ l l l l }
    \textbf{Component} & \textbf{Proof Size (bytes)} & \textbf{Prover Time (ms)} & \textbf{Verifier Time (ms)} \\ \hline
    \textbf{Signature Verification (7x)} & 316888 & 4986.8 & 2274.0 \\
    \textbf{Rest of Proof} & 16328 & 30.6 & 36.2 \\
    \textbf{Total} & 333216 & 5017.4 & 2310.2 \\
    \end{tabular}
    \caption{Running times and proof size for our zero knowledge proof of exclusion using the signatures of \cite{CL02}. 
}
\label{ZKP_perf}
\end{center}
\end{figure*}

\begin{figure}
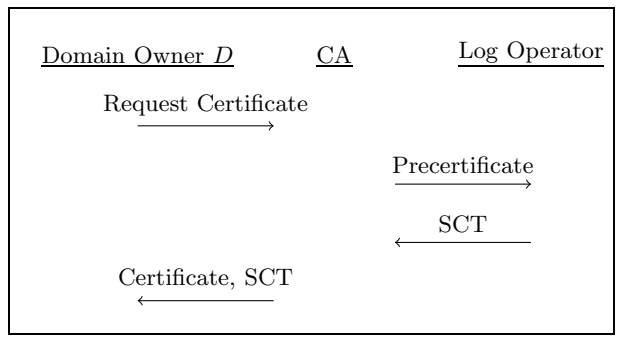

	\centering
	    \small
    \begin{tabular}{ l l l }
    \textbf{Signature} & \textbf{CL02\cite{CL02}} & \textbf{BBS\cite{BBS04}} \\ \hline
    \textbf{Log Entry Growth (bytes)} & 1791 & 480 \\
    \textbf{SCT Growth (bytes)} & 597 & 160 \\ 
    \end{tabular}
    \caption{Expected growth in log entry and SCT size due to inclusion of signatures required for our zero knowledge proof of exclusion protocol when using RSA-based CL02 signatures or pairing-based BBS signatures. Log entries require including 3 signatures, and SCTs require 1.}
\label{ZKP_bloat}
\end{figure}
We built a prototype of our protocol in parts using ZKPDL \cite{ZKPDL} for the purpose of estimating the performance of such a scheme and ran it on a consumer laptop equipped with an Intel Core i5-2540M Sandy Bridge Processor at 2.60GHz. Our implementation consisted of the entire protocol except the verification of the signatures, for which a ZKPDL implementation already exists for the RSA-based CL02 signature scheme of Camenisch and Lysyanskaya \cite{CL02}. Using 2048-bit RSA keys, we measured the proof size, prover computation time, and verifier computation time for our implementation as well as the signature verification (which we multiply by 7 to account for the 7 verifications needed in the proof). The numbers shown in figure \ref{ZKP_perf} are averages over 20 executions. Although the protocol may seem costly, we point out that it is meant to be executed infrequently to catch cheating logs and is not expected to be part of regular, day-to-day web browsing activities. As such, the cost of under 10 seconds and 350\;KB on the network seems reasonable. It is important to note that the 333\;KB proof is not something that will be included in SCTs or logs but data that is generated when the proof is needed. The overhead for instantiating our scheme with CL02 and BBS signatures is shown in figure \ref{ZKP_bloat} in terms of growth of SCTs and log entries. These numbers, especially those for BBS signatures, are very reasonable when it is considered that the SCTs will typically be delivered along with certificate chains that are already sometimes several kilobytes large.

\begin{figure}[h]
\centering
\begin{framed}
\begin{tikzpicture}[
block/.style={align=left, minimum width=1.6cm, minimum height=.5em},
space/.style={align=left, minimum width=0cm}]

\node (D) [block] {\small \underline{Domain Owner $D$}};
\node (D_1) [below=1.5em of D, block]{};
\node (D_2) [below=1.5emof D_1, block]{};
\node (D_3) [below=1.5emof D_2, block]{}; 
\node (D_4) [below=1.5emof D_3, block]{};

\node (CA_1) [right=of D_1, block]{};
\node (CA) [block, above=1.5em of CA_1]{\small \underline{CA}};
\node (CA_2) [right=of D_2, block]{};
\node (CA_3) [right=of D_3, block]{};
\node (CA_4) [right=of D_4, block]{};

\node (Log_1) [right=of CA_1, block]{};
\node (Log) [block, above=1.5em of Log_1]{\small \underline{Log Operator}};
\node (Log_2) [right=of CA_2, block]{};
\node (Log_3) [right=of CA_3, block]{};
\node (Log_4) [right=of CA_4, block]{};

\draw[->] (D_1.center) to node[above=.05em]{\small Request Certificate} (CA_1.west);
\draw[->] (CA_2.east) to node[above=.05em]{\small Precertificate} (Log_2.center);
\draw[->] (Log_3.center) to node[above=.05em]{\small SCT} (CA_3.east);
\draw[->] (CA_4.west) to node[above=.05em]{\small Certificate, SCT} (D_4.center);

\end{tikzpicture}
\end{framed}
\begin{framed}
\begin{tikzpicture}[
block/.style={align=left, minimum width=2cm, minimum height=.5em},
space/.style={align=left, minimum width=.1cm}]

\node (D_1) [block]{};
\node (D) [block, above=1em of D_1] {\small \underline{Domain Owner $D$}};
\node (D_2) [below=2emof D_1, block]{};
\node (space_0) [right=of D, space]{};
\node (space_1) [right=of D_1, space]{};
\node (space_2) [right=of D_2, space]{};

\node (V_1) [right=of space_1, block]{};
\node (V) [block, above=1emof V_1]{\small \underline{Site Visitor}};
\node (V_2) [right=of space_2, block]{};
\draw[->] (V_1.west) to node[above=.05em]{\small Request Site} (D_1.east);
\draw[->] (D_2.east) to node[above=.05em]{\small Certificate, SCT} (V_2.west);

\end{tikzpicture}
\end{framed}
\caption{A high-level view of interactions between parties involved generating and retrieving certificates in CT. We will modify the messages sent in each message to hide private subdomains.}
\label{privCT}
\end{figure}
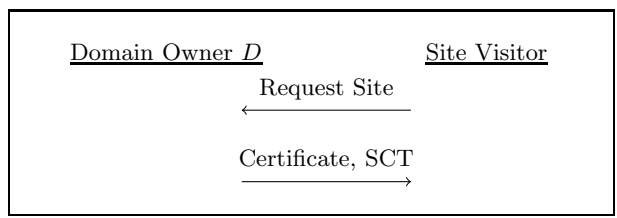

\section{Private Subdomains in CT}
\label{sec:privateCT}

We next turn to the second privacy difficulty discussed in
the introduction: extending CT to handle private subdomains, such as
the internal subdomains of an enterprise network.  
Recall that CT requires every SCT and log 
entry to reveal the domain name for which the
associated certificate is issued. 
This prevents companies who have private internal subdomains 
from using a public CA,
an undesirable consequence of CT. 

A great deal of discussion has taken place around
this issue in the CT forum~\cite{Redaction, redaction-disc}.  The main proposed solutions 
involve redaction of names via wildcards or inclusion of a
domain name hash instead of a cleartext name.  Wildcards in log entries pose
a security risk and can only be used to redact the leftmost label of a domain name. Hashed domain names, 
on the other hand, are vulnerable to a dictionary attack. Since domain names tend to have very little 
entropy, a dictionary attack is a significant threat to a solution that relies on hashing.
 For these reasons, the current draft RFC for CT version~2.0
has no solution for this problem~\cite{RFC-bis}.   

In this section, we will consider the various requirements of a solution to this problem and establish a clear threat model before providing a solution which protects against all the threats we enumerate.

\subsection{Threat Model}
Figure \ref{privCT} depicts interactions between all parties involved in the issuance and use of a certificate in CT. The domain owner requests a certificate from a CA who sends a precertificate to a log, which in turn sends the CA an SCT that is passed on, along with a certificate, to the domain owner. Later, a site visitor or auditor visits the site for which the certificate was issued and receives the certificate and the SCT. We will modify the information sent in each message in order to hide private subdomains. 

We state the security requirements from the standpoint of four parties involved in this scheme:
\begin{itemize}
\item[--]\textbf{Outsider:} An outsider viewing the log without having visited a private subdomain cannot learn anything about the private subdomain except its suffix.  For example, the SCT for {\tt private.company.com} should reveal nothing beyond the fact that it belongs to a subdomain of {\tt company.com}. 
 
\item[--]\textbf{Domain Owner:} The domain owner must not be able to frame an honest log by claiming that it has an SCT for a certificate that does not appear on the log. 
 
\item[--]\textbf{CA:} The CA must not be able to issue a new, valid certificate using an SCT for another domain that already appears in a log. 

\item[--]\textbf{Log:} Logs have the same requirements as in CT without private domains. Auditors and monitors should be able to verify that SCTs they encounter appear in logs, and monitors should be able to examine the log and find any fraudulent certificates therein. 
\end{itemize}

\subsection{Private Subdomains}\label{sec:privdom}
These requirements can be met with a scheme that uses only cryptographic commitments that guarantee binding and hiding properties. In this scheme, the name of a domain owner and a commitment to a subdomain appear in CT logs instead of the subdomain itself. The decommitment to the subdomain acts as a proof that the site gives to visitors, allowing them to verify that the SCT belongs to their domain and that it appears in a log. 
 
 The construction is as follows (summarized in figure \ref{privSub}): The domain owner $D$ generates a commitment $C_d$ to subdomain $d$ (e.g. for the domain name {\tt secret.example.com}, we have $D=${\tt example.com} and $d=${\tt secret}) with decommitment randomness $r$.  The domain owner sends $(d, D, C_d, r)$ to the CA.   The CA checks that $C_d$ is computed correctly and passes $(D, C_d)$ to the log in the signed precertificate. The log incorporates $(D, C_d)$ into the log entry and SCT and sends the SCT back to the CA who passes it on to the domain owner as usual. $d$ and $r$ are embedded into the final certificate issued by the CA.  Now any visitor to the site (or any auditor/monitor) is given the certificate and the SCT.  The visitor verifies that the commitment $C_d$ in the SCT is in fact a commitment to $d$ with decommitment randomness $r$.  Monitors who audit the logs can check that the correct number of {\tt example.com} certificates are present on the logs and that no spurious certificates have been issued.  \\

\begin{figure*}
\centering
\begin{framed}
\begin{tikzpicture}[
block/.style={align=left, minimum width=4cm, minimum height=3em},
tallblock/.style={align=left, minimum width=4cm, minimum height=7em},
textblock/.style={align=left, minimum width=4cm, minimum height=1em},
space/.style={align=left, minimum width=0cm}]

\node (D) [block] {\underline{Domain Owner $D$}};
\node (D_text1) [below=0em of D, block]{$C =$ commit(``{\tt secret}'', $r$)};
\node (D_1) [below=1.5em of D_text1, block]{};
\node (D_2) [below=1.5emof D_1, block]{};
\node (D_3) [below=1.5emof D_2, block]{}; 
\node (D_4) [below=1.5emof D_3, tallblock]{};

\node (CA_1) [right=of D_1, block]{};
\node (CA_text1) [right=of D_text1, block]{};
\node (CA) [block, above=0em of CA_text1]{\underline{CA}};
\node (CA_2) [right=of D_2, block]{};
\node (CA_3) [right=of D_3, block]{};
\node (CA_4) [right=of D_4, block]{};

\node (Log_1) [right=of CA_1, block]{};
\node (Log_text1) [right=of CA_text1, block]{};
\node (Log) [block, above=0em of Log_text1]{\underline{Log Operator}};
\node (Log_2) [right=of CA_2, block]{};
\node (Log_3) [right=of CA_3, block]{};
\node (Log_4) [right=of CA_4, block]{log($C$, {\tt example.com})};

\draw[->] (D_1.center) to node[above=.05em, align=center]{Request Certificate\\``{\tt secret}'', ``{\tt example.com}'', $C$, $r$} (CA_1.center);
\draw[->] (CA_2.center) to node[above=.05em, align=center]{Precertificate\\$C$, ``{\tt example.com}''} (Log_2.center);
\draw[->] (Log_3.center) to node[above=.05em, align=center]{SCT \\$C$, ``{\tt example.com}''} (CA_3.center);
\draw[->] (CA_4.center) to node[above=.05em, align=center]{Certificate:\\ ``{\tt secret.example.com}'', $r$ \\[.3\baselineskip]SCT: \\$C$, ``{\tt example.com}''} (D_4.center);

\end{tikzpicture}
\end{framed}
\begin{framed}
\begin{tikzpicture}[
block/.style={align=left, minimum width=4cm, minimum height=5em},
textblock/.style={align=left, minimum width=4cm, minimum height=1em},
space/.style={align=left, minimum width=.1cm}]

\node (D_1) [block]{};
\node (D) [block, above=1em of D_1] {\underline{Domain Owner $D$}};
\node (D_2) [below=2emof D_1, block]{};
\node (space_0) [right=of D, space]{};
\node (space_1) [right=of D_1, space]{};
\node (space_2) [right=of D_2, space]{};

\node (V_1) [right=of space_1, block]{};
\node (V) [block, above=1emof V_1]{\underline{Site Visitor}};
\node (V_2) [right=of space_2, block]{};
\node (V_3) [below=1em of V_2, textblock]{verify($C$, {\tt secret}, $r$)};
\draw[->] (V_1.center) to node[above=.05em, align=center]{Page Request:\\ ``{\tt secret.example.com}''} (D_1.center);
\draw[->] (D_2.center) to node[above=.05em, align=center]{Certificate:\\ ``{\tt secret.example.com}'', $r$ \\[.3\baselineskip]SCT: \\$C$, ``{\tt example.com}''} (V_2.center);

\end{tikzpicture}
\end{framed}
\caption{CT augmented with support for private subdomains. $C$ represents a commitment to the private subdomain ``{\tt secret}'' with decommitment randomness $r$. The SCT and log entry for ``{\tt secret.example.com}'' replace the private subdomain with $C$, and visitors to the site verify that the commitment corresponds to the appropriate subdomain.}
\label{privSub}
\end{figure*}

\subsection{Security}
We now argue that this scheme achieves security in the threat model discussed above. Below we outline how each of the security requirements imposed by private subdomains are met.\\

\noindent \textbf{Outsider:} Due to the hiding property of the commitment, an outsider observing log entries does not learn the subdomain corresponding to $C_d$ from the log.\\

\noindent \textbf{Domain Owner:} Domain owners cannot frame an honest log because every SCT from that log is signed by the log. \\

\noindent \textbf{CA:} The binding property of the commitment scheme prevents a CA from issuing fraudulent certificates that reuse existing SCTs because doing so would require finding, for a commitment $C_d$ to a subdomain $d$, an alternative decommitment to $d'$ for which the fraudulent certificate would be issued. \\

\noindent \textbf{Log:} The means CT provides for auditing and exposing misbehaving logs are not affected by private subdomains. We point out that, as mentioned above, any auditor or monitor given access to $d$ and $r$ can verify an SCT just as well as it could an SCT for a non-private domain. Moreover, a monitor working on behalf of a domain owner will still be able to detect fraudulent certificates issued for the domain owner's private subdomains if it is informed of the number of private subdomains the domain owner has registered. If the number of such subdomains appearing in logs differs from the expected number, then it is clear that either there are extra, fraudulent certificates or that the log has withheld certificates for which it issued SCTs. 

\subsection{Practical Considerations}
The SCT from the log, along with $d$ and $r$, can be embedded in the final certificate issued by the CA.  A browser visiting the site would first validate the certificate and the signature on the SCT.  It would then extract $r$ and $d$ from the certificate and use them to verify the commitment $C_d$ in the SCT.

\section{Short-Lived Certificates in CT}
\label{sec:shortLived}

We also consider the interaction of Certificate Transparency with another proposed improvement to the web's public-key infrastructure: short-lived certificates \cite{Riv98}. The fundamental idea behind short-lived certificates is to issue certificates with relatively short validity periods (i.e. one or a few days) in order to remove the need on the part of a client to check whether or not a certificate received remains valid. Instead of revoking a certificate, the domain owner in the short-lived certificate setting simply refrains from issuing a new short-lived certificate in the event of key compromise. Topalovic et al. \cite{TSH+12} explore the details of this idea and compare it favorably to other proposed and implemented solutions for certificate revocation.  

Considering the possibility of applying short-lived certificates in conjunction with certificate transparency presents the new difficulty of avoiding bloated log sizes. The frequency with which new certificates would need to be issued under the short-lived certificate regime would mean that where one log entry would suffice for a regular certificate, nearly 100 short-lived certificates would be needed. A naive integration of the two solutions would lead to log sizes far larger than is truly necessary, but a simple solution can remedy this issue and allow the two solutions to work together quite well. 

Instead of creating one log entry per certificate for short-lived certificates, a large number of potential short-lived certificates will be allotted one log entry. This log entry will have a special flag set to indicate that it corresponds to a family of short-lived certificates, and the validity period for the log entry will be comparable to that of a regular, long-lived certificate. The special log entry will also include the root of a Merkle tree of all the short-lived certificates affiliated with the entry. When visiting a site that uses short-lived certificates, auditors will receive a proof that the SCT for that site's certificate is in the Merkle tree whose root appears in the corresponding log entry. This is the best-of-both-worlds: the growth rate of the CT log is no different from that due to a regular certificate, but the short-lived nature of the individual certificates also resolves revocation issues.

\section{Related Work}

Well-documented concerns about the state of the web's public-key infrastructure (PKI) \cite{CO13} have led to a number of proposed solutions for improving the capacity to catch problematic CAs or to dispense with CAs altogether. Pinning solutions such as Tack \cite{TACK} improve security by adding a ``pinned'' signing key that is associated with a server and without which the server will not be considered authentic. In addition to CT, a handful of other proposed PKI improvements rely on the concept of transparent logs. The EFF's Sovereign Keys proposal uses a log to list the Sovereign Keys associated with domains and allows for automatically routing around certain certficate-related attacks where users tend to click through browser warnings \cite{SovKeys}. AKI and ARPKI implement a set of checks and balances to prevent problems rooted in the misbehavior of one or a few CAs, and ARPKI is co-designed with a formal model to prove its security properties \cite{KHP+13, BCK+14}. Namecoin uses a blockchain to achieve goals similar to those of other improved PKI schemes \cite{Namecoin}. Taking a somewhat different approach, Etemad and Kupcu \cite{EK15} propose a scheme where non-colluding servers gossip to catch any cheating parties.

Certificate Transparency has received heightened academic attention as it begins to be widely deployed in Chrome and Chromium browsers. Multiple works have analyzed the security of Certificate Transparency and its generalizations or designed efficient protocols for parts of the scheme that are not specified in the current RFC. Chuat et al. \cite{CSP+15} provide efficient gossip protocols to ensure that different users of CT do in fact have the same view of the logs with which they interact, and Dowling et al. \cite{DGHS16} formally prove the security of many aspects of Certificate transparency. In the space of extensions and modifications to CT, Singh et al \cite{SSR17} propose a variant of CT with extended functionality and shorter log proofs and Syta et al \cite{STV+16} propose CoSi, a ``witness cosigning'' protocol that removes the need for gossip in CT by ensuring that every certificate is validated by many witnesses, offering protection even against sustained man in the middle attacks. Chase and Meiklejohn \cite{CM16} prove the security of a generalization of Certificate Transparency they call ``Transparency Overlays,'' which can be added on top of other protocols such as Bitcoin to add a layer of transparency. They leave the exploration of the interaction between transparency and privacy as an interesting open problem, one which our work begins to study.

The notion of transparency itself has been studied in various settings not necessarily related to the web's PKI and appears to be a fairly general notion. Revocation Transparency provides a transparency-based solution for certificate revocation \cite{RevTrans}. DECIM \cite{YRK15} uses a transparency log to keep track of uses of a public key in a messaging context, thereby enabling its owner to detect key misuse. Fahl et al. \cite{FDP+14} improve the security of mobile app distribution with ``Application Transparency.''  Extended Certificate Transparency \cite{Rya14} adds a second log to efficiently handle certificate revocation and presents an application to secure email. Finally, a number of works have focused on building general transparent data structures and enabling computation over them \cite{Insynd, Balloon, miller2014authenticated, Versum}. 

At least two systems have combined transparency with privacy notions. CONIKS~\cite{MBB+15} provides a transparent identity/value map (with key management for secure communication as a motivating application) while hiding identifiers and values. NSEC5~\cite{GNPR16} is an extension to DNSSEC providing stronger privacy by hiding the set of valid subdomains for any given domain. Both use VRFs as a key building block for preserving privacy.

\section{Availability}

The simulation code for the second part of our zero knowledge proof of exclusion can be found on GitHub at \url{https://github.com/SabaEskandarian/CTZKPExcl}. The code used to interpret and run the simulation, along with the implementation of the signature scheme is part of the ZKPDL project \cite{ZKPDL} at \url{https://github.com/brownie/cashlib}.

\section{Conclusion and Open Questions}
We have examined various elements of Certificate Transparency with the goal of applying cryptographic solutions to remedy the compromise of private information without significantly disrupting the operation of CT as it is currently designed. Our primary goals were to develop a privacy-preserving means of exposing logs that have excluded certificates after issuing SCTs for them and to securely redact private subdomains in CT logs. We accomplished the former goal by developing a zero knowledge proof that achieves the best possible security for privately implicating a cheating log. The latter problem was solved by the realization that commitment schemes appropriately applied exactly match the security properties required for private subdomains. We also showed how CT and short-lived certificates can be used in concert. It is our hope that these practical solutions based on widely-used cryptographic primitives can be applied to promote a more transparent web without any compromise in user privacy. 

We leave the following open questions for future work:
\begin{itemize}
\item Are there improved, practical schemes for browsers to privately query CT logs? We considered some partial solutions in section \ref{ZKpractical}, but a complete solution to this problem would be very useful. One possibility is to use a mechanism similar to OCSP stapling, where the proof of SCT inclusion in the log is sent with a site's certificate, but this would require more bandwidth and frequent updates to the proof to use an up-to-date STH. 
\item Can the zero knowledge proof of exclusion of log entries we present be modified to require fewer signatures or even more lightweight primitives, thereby reducing overhead in SCTs and log entries?
\item We have shown how to achieve private subdomains in CT. Is it also possible to efficiently make entire domain names private? Such a scheme would be useful in practice for enterprises who wish to register domains for new projects before announcing them publicly. 

\end{itemize}

\section*{Acknowledgements}

This work is supported by grants from NSF, DARPA, ONR, 
and the Simons Foundation. 
Opinions, findings and conclusions or recommendations expressed in this
material are those of the authors and do not necessarily reflect the views 
of DARPA.

\bibliographystyle{plain}
\bibliography{paper}

\begin{thebibliography}{10}

\bibitem{redaction-disc}
Certificate transparency policy (google groups).
\newblock \url{groups.google.com/a/chromium.org/forum/#!forum/ct-policy}.

\bibitem{DigiNotar}
Heather Adkins.
\newblock An update on attempted man-in-the-middle attacks, 2011.
\newblock
  \url{security.googleblog.com/2011/08/update-on-attempted-man-in-middle.html}.

\bibitem{basin2014arpki}
David Basin, Cas Cremers, Tiffany Hyun-Jin Kim, Adrian Perrig, Ralf Sasse, and
  Pawel Szalachowski.
\newblock Arpki: Attack resilient public-key infrastructure.
\newblock In {\em Proceedings of the 2014 ACM SIGSAC Conference on Computer and
  Communications Security}, pages 382--393. ACM, 2014.

\bibitem{BCK+14}
David~A. Basin, Cas J.~F. Cremers, Tiffany~Hyun{-}Jin Kim, Adrian Perrig, Ralf
  Sasse, and Pawel Szalachowski.
\newblock {ARPKI:} attack resilient public-key infrastructure.
\newblock In {\em Proceedings of the 2014 {ACM} {SIGSAC} Conference on Computer
  and Communications Security, Scottsdale, AZ, USA, November 3-7, 2014}, pages
  382--393, 2014.

\bibitem{BBS04}
Dan Boneh, Xavier Boyen, and Hovav Shacham.
\newblock Short group signatures.
\newblock In {\em Advances in Cryptology - {CRYPTO} 2004, 24th Annual
  International CryptologyConference, Santa Barbara, California, USA, August
  15-19, 2004, Proceedings}, pages 41--55, 2004.

\bibitem{Bou00}
Fabrice Boudot.
\newblock Efficient proofs that a committed number lies in an interval.
\newblock In {\em Advances in Cryptology - {EUROCRYPT} 2000, International
  Conference on the Theory and Application of Cryptographic Techniques, Bruges,
  Belgium, May 14-18, 2000, Proceeding}, pages 431--444, 2000.

\bibitem{Bra97}
Stefan Brands.
\newblock Rapid demonstration of linear relations connected by boolean
  operators.
\newblock In {\em Advances in Cryptology - {EUROCRYPT} '97, International
  Conference on the Theory and Application of Cryptographic Techniques,
  Konstanz, Germany, May 11-15, 1997, Proceeding}, pages 318--333, 1997.

\bibitem{CCS08}
Jan Camenisch, Rafik Chaabouni, and Abhi Shelat.
\newblock Efficient protocols for set membership and range proofs.
\newblock In {\em Advances in Cryptology - {ASIACRYPT} 2008, 14th International
  Conference on the Theory and Application of Cryptology and Information
  Security, Melbourne, Australia, December 7-11, 2008. Proceedings}, pages
  234--252, 2008.

\bibitem{CL02}
Jan Camenisch and Anna Lysyanskaya.
\newblock A signature scheme with efficient protocols.
\newblock In {\em Security in Communication Networks, Third International
  Conference, {SCN} 2002, Amalfi, Italy, September 11-13, 2002. Revised
  Papers}, pages 268--289, 2002.

\bibitem{CL04}
Jan Camenisch and Anna Lysyanskaya.
\newblock Signature schemes and anonymous credentials from bilinear maps.
\newblock In {\em Advances in Cryptology - {CRYPTO} 2004, 24th Annual
  International CryptologyConference, Santa Barbara, California, USA, August
  15-19, 2004, Proceedings}, pages 56--72, 2004.

\bibitem{CM01}
Jan Camenisch and Markus Michels.
\newblock Proving in zero-knowledge that a number is the product of two safe
  primes.
\newblock In {\em Advances in Cryptology - {EUROCRYPT} '99, International
  Conference on the Theory and Application of Cryptographic Techniques, Prague,
  Czech Republic, May 2-6, 1999, Proceeding}, pages 107--122, 1999.

\bibitem{CM16}
Melissa Chase and Sarah Meiklejohn.
\newblock Transparency overlays and applications.
\newblock In {\em Proceedings of the 2016 {ACM} {SIGSAC} Conference on Computer
  and Communications Security, Vienna, Austria, October 24-28, 2016}, pages
  168--179, 2016.

\bibitem{CEvdG87}
David Chaum, Jan{-}Hendrik Evertse, and Jeroen van~de Graaf.
\newblock An improved protocol for demonstrating possession of discrete
  logarithms and some generalizations.
\newblock In {\em Advances in Cryptology - {EUROCRYPT} '87, Workshop on the
  Theory and Application of of Cryptographic Techniques, Amsterdam, The
  Netherlands, April 13-15, 1987, Proceedings}, pages 127--141, 1987.

\bibitem{CP93}
David Chaum and Torben~P. Pedersen.
\newblock Wallet databases with observers.
\newblock In {\em Advances in Cryptology - {CRYPTO} '92, 12th Annual
  International Cryptology Conference, Santa Barbara, California, USA, August
  16-20, 1992, Proceedings}, pages 89--105, 1992.

\bibitem{CSP+15}
Laurent Chuat, Pawel Szalachowski, Adrian Perrig, Ben Laurie, and Eran Messeri.
\newblock Efficient gossip protocols for verifying the consistency of
  certificate logs.
\newblock In {\em 2015 {IEEE} Conference on Communications and Network
  Security, {CNS} 2015, Florence, Italy, September 28-30, 2015}, pages
  415--423, 2015.

\bibitem{CO13}
Jeremy Clark and Paul~C. van Oorschot.
\newblock Sok: {SSL} and {HTTPS:} revisiting past challenges and evaluating
  certificate trust model enhancements.
\newblock In {\em 2013 {IEEE} Symposium on Security and Privacy, {SP} 2013,
  Berkeley, CA, USA, May 19-22, 2013}, pages 511--525, 2013.

\bibitem{DBB+15}
Gaby~G. Dagher, Benedikt B{\"{u}}nz, Joseph Bonneau, Jeremy Clark, and Dan
  Boneh.
\newblock Provisions: Privacy-preserving proofs of solvency for bitcoin
  exchanges.
\newblock In {\em Proceedings of the 22nd {ACM} {SIGSAC} Conference on Computer
  and Communications Security, Denver, CO, USA, October 12-6, 2015}, pages
  720--731, 2015.

\bibitem{DF01}
Ivan Damg{\aa}rd and Eiichiro Fujisaki.
\newblock An integer commitment scheme based on groups with hidden order.
\newblock {\em {IACR} Cryptology ePrint Archive}, 2001:64, 2001.

\bibitem{DGHS16}
Benjamin Dowling, Felix G{\"{u}}nther, Udyani Herath, and Douglas Stebila.
\newblock Secure logging schemes and certificate transparency.
\newblock In {\em Computer Security - {ESORICS} 2016 - 21st European Symposium
  on Research in Computer Security, Heraklion, Greece, September 26-30, 2016,
  Proceedings, Part {II}}, pages 140--158, 2016.

\bibitem{eckersley2011sovereign}
Peter Eckersley.
\newblock Sovereign keys: A proposal to make https and email more secure.
\newblock {\em Electronic Frontier Foundation}, 18, 2011.

\bibitem{EK15}
Mohammad Etemad and Alptekin K{\"u}p{\c{c}}{\"u}.
\newblock {\em Efficient Key Authentication Service for Secure End-to-End
  Communications}, pages 183--197.
\newblock Springer International Publishing, Cham, 2015.

\bibitem{HPKP}
C.~Evans, C.~Palmer, and R.~Sleevi.
\newblock Public key pinning extension for http.
\newblock RFC 7469, April 2015.

\bibitem{FDP+14}
Sascha Fahl, Sergej Dechand, Henning Perl, Felix Fischer, Jaromir Smrcek, and
  Matthew Smith.
\newblock Hey, {NSA:} stay away from my market! future proofing app markets
  against powerful attackers.
\newblock In {\em Proceedings of the 2014 {ACM} {SIGSAC} Conference on Computer
  and Communications Security, Scottsdale, AZ, USA, November 3-7, 2014}, pages
  1143--1155, 2014.

\bibitem{SovKeys}
Electronic~Frontier Foundation.
\newblock Sovereign keys.
\newblock \url{www.eff.org/sovereign-keys}.

\bibitem{GNPR16}
Sharon Goldberg, Moni Naor, Dimitrios Papadopoulos, and Leonid Reyzin.
\newblock {NSEC5} from elliptic curves: Provably preventing {DNSSEC} zone
  enumeration with shorter responses.
\newblock {\em {IACR} Cryptology ePrint Archive}, 2016:83, 2016.

\bibitem{DANE}
P.~Hoffman and J.~Schlyter.
\newblock The dns-based authentication of named entities (dane) transport layer
  security (tls) protocol: Tlsa.
\newblock RFC 6698, August 2012.

\bibitem{kasten2013cage}
James Kasten, Eric Wustrow, and J~Alex Halderman.
\newblock Cage: Taming certificate authorities by inferring restricted scopes.
\newblock In {\em International Conference on Financial Cryptography and Data
  Security}, pages 329--337. Springer, 2013.

\bibitem{KHP+13}
Tiffany~Hyun{-}Jin Kim, Lin{-}Shung Huang, Adrian Perrig, Collin Jackson, and
  Virgil~D. Gligor.
\newblock Accountable key infrastructure {(AKI):} a proposal for a public-key
  validation infrastructure.
\newblock In {\em 22nd International World Wide Web Conference, {WWW} '13, Rio
  de Janeiro, Brazil, May 13-17, 2013}, pages 679--690, 2013.

\bibitem{kim2013accountable}
Tiffany Hyun-Jin Kim, Lin-Shung Huang, Adrian Perring, Collin Jackson, and
  Virgil Gligor.
\newblock Accountable key infrastructure (aki): a proposal for a public-key
  validation infrastructure.
\newblock In {\em Proceedings of the 22nd international conference on World
  Wide Web}, pages 679--690. ACM, 2013.

\bibitem{RevTrans}
B.~Laurie and E.~Kasper.
\newblock Revocation transparency.
\newblock \url{www.links.org/files/RevocationTransparency.pdf}.

\bibitem{RFC}
B.~Laurie, A.~Langley, and E.~Kasper.
\newblock Certificate transparency.
\newblock RFC 6962, June 2013.

\bibitem{RFC-bis}
B.~Laurie, A.~Langley, E.~Kasper, E.~Messeri, and R.~Stradling.
\newblock Certificate transparency version 2.0.
\newblock RFC-bis 6962-bis, 2017.

\bibitem{Comodo}
Ben Laurie.
\newblock Improving ssl certificate security, 2011.
\newblock
  \url{security.googleblog.com/2011/04/improving-ssl-certificate-security.html}.

\bibitem{TACK}
Moxie Marlinspike and Trevor Perrin.
\newblock Trust assertions for certificate keys.
\newblock \url{tack.io/draft.html}, 2013.

\bibitem{ZKPDL}
Sarah Meiklejohn, C.~Christopher Erway, Alptekin K{\"{u}}p{\c{c}}{\"{u}},
  Theodora Hinkle, and Anna Lysyanskaya.
\newblock {ZKPDL:} {A} language-based system for efficient zero-knowledge
  proofs and electronic cash.
\newblock In {\em 19th {USENIX} Security Symposium, Washington, DC, USA, August
  11-13, 2010, Proceedings}, pages 193--206, 2010.

\bibitem{MBB+15}
Marcela~S. Melara, Aaron Blankstein, Joseph Bonneau, Edward~W. Felten, and
  Michael~J. Freedman.
\newblock {CONIKS:} bringing key transparency to end users.
\newblock In {\em 24th {USENIX} Security Symposium, {USENIX} Security 15,
  Washington, D.C., USA, August 12-14, 2015.}, pages 383--398, 2015.

\bibitem{miller2014authenticated}
Andrew Miller, Michael Hicks, Jonathan Katz, and Elaine Shi.
\newblock Authenticated data structures, generically.
\newblock In {\em ACM SIGPLAN Notices}, volume~49, pages 411--423. ACM, 2014.

\bibitem{Namecoin}
Namecoin.
\newblock \url{namecoin.org}.

\bibitem{Ped92}
Torben~P. Pedersen.
\newblock Non-interactive and information-theoretic secure verifiable secret
  sharing.
\newblock In {\em Advances in Cryptology - {CRYPTO} '91, 11th Annual
  International Cryptology Conference, Santa Barbara, California, USA, August
  11-15, 1991, Proceedings}, pages 129--140, 1991.

\bibitem{Insynd}
Roel Peeters and Tobias Pulls.
\newblock Insynd: Improved privacy-preserving transparency logging.
\newblock In {\em Computer Security - {ESORICS} 2016 - 21st European Symposium
  on Research in Computer Security, Heraklion, Greece, September 26-30, 2016,
  Proceedings, Part {II}}, pages 121--139, 2016.

\bibitem{Balloon}
Tobias Pulls and Roel Peeters.
\newblock Balloon: {A} forward-secure append-only persistent authenticated data
  structure.
\newblock In {\em Computer Security - {ESORICS} 2015 - 20th European Symposium
  on Research in Computer Security, Vienna, Austria, September 21-25, 2015,
  Proceedings, Part {II}}, pages 622--641, 2015.

\bibitem{Riv98}
Ronald~L. Rivest.
\newblock Can we eliminate certificate revocations lists?
\newblock In {\em Financial Cryptography, Second International Conference,
  FC'98, Anguilla, British West Indies, February 23-25, 1998, Proceedings},
  pages 178--183, 1998.

\bibitem{Rya14}
Mark~Dermot Ryan.
\newblock Enhanced certificate transparency and end-to-end encrypted mail.
\newblock In {\em 21st Annual Network and Distributed System Security
  Symposium, {NDSS} 2014, San Diego, California, USA, February 23-26, 2014},
  2014.

\bibitem{Sch91}
Claus{-}Peter Schnorr.
\newblock Efficient signature generation by smart cards.
\newblock {\em J. Cryptology}, 4(3):161--174, 1991.

\bibitem{SSR17}
Abhishek Singh, Binanda Sengupta, and Sushmita Ruj.
\newblock Certificate transparency with enhancements and short proofs.
\newblock 2017.

\bibitem{Redaction}
R.~Stradling and E.~Messeri.
\newblock Certificate transparency: Domain label redaction.
\newblock Internet-draft, 2017.

\bibitem{STV+16}
Ewa Syta, Iulia Tamas, Dylan Visher, David~Isaac Wolinsky, Philipp Jovanovic,
  Linus Gasser, Nicolas Gailly, Ismail Khoffi, and Bryan Ford.
\newblock Keeping authorities "honest or bust" with decentralized witness
  cosigning.
\newblock In {\em {IEEE} Symposium on Security and Privacy, {SP} 2016, San
  Jose, CA, USA, May 22-26, 2016}, pages 526--545, 2016.

\bibitem{TSH+12}
Emin Topalovic, Brennan Saeta, Lin shung Huang, Collin Jackson, and Dan Boneh.
\newblock Towards short-lived certificates.
\newblock In {\em W2SP}, 2012.

\bibitem{Versum}
Jelle van~den Hooff, M.~Frans Kaashoek, and Nickolai Zeldovich.
\newblock Versum: Verifiable computations over large public logs.
\newblock In {\em Proceedings of the 2014 {ACM} {SIGSAC} Conference on Computer
  and Communications Security, Scottsdale, AZ, USA, November 3-7, 2014}, pages
  1304--1316, 2014.

\bibitem{perspectives}
Dan Wendlandt, David~G. Andersen, and Adrian Perrig.
\newblock \emph{Perspectives: } improving ssh-style host authentication with
  multi-path probing.
\newblock In {\em 2008 {USENIX} Annual Technical Conference, Boston}, pages
  321--334, 2008.

\bibitem{WoSign}
Andrew Whalley.
\newblock Distrusting wosign and startcom certificates, 2016.
\newblock
  \url{security.googleblog.com/2016/10/distrusting-wosign-and-startcom.html}.

\bibitem{YRK15}
J.~Yu, M.~Ryan, and C.~Kremers.
\newblock Decim: Detecting endpoint compromise in messaging.
\newblock {\em {IACR} Cryptology ePrint Archive}, 2015, 2015.

\end{thebibliography}

\appendix

\section{Security Model and Proof Details}\label{ZKApp}
Below we give a formalization of the definition of a CT log that has been modified to include our additional requirements. 
\begin{definition}[CT Log]
A CT log is an entity that maintains append-only lists Log, hashLog, ILog, and TLog as well as 4 signing keys $k_H, k_T, k_I, k_S$ and a counter $i$. A CT log's lists are initially empty, and the counter is initially set to 0. It also has access to $\textit{poly}(\lambda)$-near collision resistant hash function $H$, a signature scheme $\sigma=(\textit{KeyGen}', \textit{Sign}', \textit{Verify}')$ , and another signature scheme $\sigma_k = (\textit{KeyGen}, Sign, \textit{Verify})$ with efficient proof of knowledge properties (as in Definition~\ref{def:sigcom}). Upon receiving a message $m, t$ from party $P$, a log forms the SCT $s=(m,t)$ and takes the following actions\footnote{Although in reality it is the log who assigns $t$ and not $P$, defining the security model such that $P$ selects $t$ makes the adversary strictly more powerful in our setting.}:
\begin{itemize}
\item send $P$ $(H(s), \textit{Sign}_{k_H}(H(s)))$
\item send $P$ $(t+H(s), \textit{Sign}_{k_T}(t+H(s)))$
\item send $P$ the SCT $s$ and $\textit{Sign}'_{k_S}(s)$
\item append $e = (m, i, t)$ to Log
\item append $(H(e), \textit{Sign}_{k_H}(H(e)))$ to hashLog
\item append $(i+H(e), \textit{Sign}_{k_I}(i+H(e)))$ to ILog
\item append $(t+H(e), \textit{Sign}_{k_T}(t+H(e)))$ to TLog
\item increment $i$ by one
\end{itemize}

A well-formed list Log consists of $\textit{poly}(\lambda)$ entries ordered by increasing order of $I$. Values of $I$ must begin at zero and be incremented by one for each subsequent entry. Values of $T$ must fall within a range of $\textit{poly}(\lambda)$ where $\lambda$ is a security parameter. Furthermore, the order of entries of $I$ and $T$ must correlate, that is, if $I_x>I_y$, then it must hold that $T_x>T_y$. Lists hashLog, ILog, and TLog have one entry per entry in Log and are sorted in the same way.

A CT log also responds to requests for entries in its lists with the requested information.
\end{definition}

We make use of the following proof of exclusion soundness game and security definition in proving the security of our proposed zero knowledge exclusion proof. 

\begin{definition}[$\textit{ProofExcl}_{\Pi,A,V,L}(\lambda)$]
The proof of exclusion game is played by an adversary $A$ with a log $L$ and Verifier $V$ using protocol $\Pi$ and security parameter $\lambda$. $L$ maintains append-only lists \emph{Log} (a well-formed CT log), \emph{hashLog}, \emph{ILog}, and \emph{TLog}, which are initially empty. $L$ also has the appropriate signing keys and a counter $i$, initialized to zero. The game consists of two phases:
\begin{enumerate}
\item $A$ interacts with $L$ in $\textit{poly}(\lambda)$ rounds where in each round $A$ sends $L$ a message $(m, t)$ with the range of all $t$'s sent being at most $\textit{poly}(\lambda)$ and each $t$ strictly greater than the previous one. After each round, $A$ retrieves from $L$ the new entry added to each list $L$ holds.
\item $A$ interacts with $V$ according to $\Pi$, with $A$ playing the role of the prover and $V$ the role of the verifier. At the end of the protocol, $V$ outputs a bit $b=1$ if it accepts the proof from $A$ as valid, and outputs $b=0$ otherwise.
\end{enumerate}
$A$ wins the game when $V$ outputs $b=1$.
\end{definition}
\begin{definition}[Zero-Knowledge Exclusion Proof]
A zero-knowledge exclusion proof is an interactive protocol between a Prover $P$ with access to log entries $x,z$, SCT $y$, signatures $\sigma_{I_x+H(x)},$ $\sigma_{T_x+H(x)},$ $\sigma_{H(x)},$ $\sigma_{T_y+H(y)},$ $\sigma_{H(y)},$ $\sigma_{I_z+H(z)},$ $\sigma_{T_z+H(z)},$ and $\sigma_{H(z)}$ and a Verifier $V$ with access to the public keys corresponding to the signatures who outputs $1$ or $0$ at the end of their interaction. Both parties are given access to a log $L$ and a security parameter $\lambda$. We require the following properties from a secure exclusion proof:
\begin{enumerate}
\item \textbf{Completeness}: $\forall x,y,z$ $s.t.$ $I_x+1=I_z, T_x<T_y<T_z$, and $\sigma_m$ is a signature on $m$ for all signatures above, $V$ outputs 1.

\item \textbf{Soundness}: $\forall$ PPT Adversaries $A$, $Pr[\textit{ProofExcl}_{\Pi,A,V,L}(\lambda)=1]\leq \textit{negl}(\lambda)$.

\item \textbf{Zero-Knowledge}: There exists an efficient simulator $S$ that, given only $\lambda$, can generate a transcript indistinguishable from that of the interaction between $P$ and $V$.
\end{enumerate}
\end{definition}

For convenience of notation, the formal proof refers to the lists of log entries, hashes of log entries, hashes plus indexes, and hashes plus timestamps as Log, hashLog, ILog, and TLog, respectively. 

\begin{theorem}
If $H$ is a $\textit{poly}(\lambda)$-near collision resistant hash function, the signature schemes $\sigma, \sigma_k$ used are existentially unforgeable, and $\sigma_k$ allows for efficient proofs of knowledge of a signature, then $\Pi$ is a secure Zero-Knowledge Exclusion Proof. 
\end{theorem}
\begin{proof}
\textbf{Completeness}. The completeness of the proof follows directly from the construction. Any prover with access to an SCT $y$ with the given properties that has access to the log can follow the protocol and convince a verifier that the log is cheating.\\
 
\textbf{Soundness}. We prove that no adversary $A$ can win the security game $\textit{ProofExcl}$ with greater than negligible probability through a series of three hybrids:
\begin{itemize}
\item \emph{hybrid 0:} The security game $\textit{ProofExcl}_{\Pi,A,V,L}(\lambda)$.
\item \emph{hybrid 1:} Same as hybrid 0, except after $V$ outputs $b$, we run the extractor for each proof of knowledge of a signature and output 0 if any of the extractions fail. This is indistinguishable from hybrid 0 by the extractability property of the proof of knowledge of a signature. Since the proof of knowledge protocol has an efficient extractor, we will never fail to extract the signatures used.
\item \emph{hybrid 2:} Same as hybrid 1, except after the verification of the extractions, the verifier checks the extracted signatures and outputs 0 if any of them are not found in the corresponding list held by $L$: either Log, hashLog, ILog, or TLog.  The adversary's advantage is at most negligibly changed from hybrid 1 due to the existential unforgeability of the signature scheme. If $V$ could ever output zero based on this additional check, we could construct a new adversary $B$ that breaks the existential unforgeability of the signature scheme by running the protocol until $V$ outputs 0 on this check and then outputting the signature $s$ that was not found in the log. Since all signatures made by $L$ are in one of the lists, $s$ must be a forged signature. Thus the additional check in this hybrid will only change the output of $V$ from hybrid 1 with negligible probability. 
\end{itemize}
Now we prove that no adversary $A$ can win hybrid game 2 with more than negligible probability. Note that we only need to prove the soundness for the first part of the protocol, as soundness for the second part of the protocol follows from the soundness of the various zero knowledge proofs it uses as subroutines. That is, we need only show that no adversary $A$ can fake the indexes or timestamps of signed entries from $L$. We show here the proof for indexes; the proof for timestamps is almost identical. 

In order for an adversary $A$ to fraudulently win the game, it must find log entries $x'$, $x''$, and $x^*$ (or $z'$, $z''$, and $z^*$) such that it can add $I_{x''}$ and $H(x')$ to get a value equal to $I_{x^*}+H(x^*)$. Since a signature on the latter value will be in ILog, this will allow $A$ to submit an acceptable proof that uses an index and hash from different log entries.

More formally, any adversary $A$ that wins the game must produce commitments $C_{I_x''+H(x')}$, $C_{I_x^*+H(x^*)}$, $C_{I_z''+H(z')}$, and $C_{I_z^*+H(z^*)}$ that satisfy the following conditions:
\begin{itemize}
\item $I_{x''}+1=I_{z''}$
\item $I_{x^*}+H(x^*)=I_{x''}+H(x')$
\item $I_{z^*}+H(z^*)=I_{z''}+H(z')$
\item $(I_{x''}+H(x'), \textit{Sign}_{k_I}(I_{x''}+H(x'))\in$ ILog
\item $(I_{z''}+H(z'), \textit{Sign}_{k_I}(I_{z''}+H(z'))\in$ ILog
\item $x'\neq x''$ OR $z'\neq z''$
\end{itemize}

We will consider the case where $x'\neq x''$. The case for $z'\neq z''$ is identical. If the adversary succeeds, then it must have found values $I_{x^*}, x^*, I_{x''}, x'$ such that 
\begin{align}I_{x^*}+H(x^*)=I_{x''}+H(x')\end{align}
Define $d=I_{x''}-I_{x*}$. Then we have that 
\begin{align}H(x^*)=H(x')+d \end{align} 
But $d < \textit{poly}(\lambda)$ since $A$ can only send polynomially many messages to $L$ to populate the log. Thus we can define adversary $B$ that breaks the $d$-near collision resistance of $H$ by calling $A$ as a subroutine, using the extractors from the zero knowledge proofs in the second part of the protocol to retrieve the contents of the commitments $A$ used, and outputting $x^*, x'$ as its near collision. This completes the soundess proof. 

\textbf{Zero-Knowledge}. The zero knowledge property is easy to establish as the sequential composition of a series of other zero knowledge proofs. The simulator starts by committing to a series of random values (except the commitment/reveal to 1, which must necessarily be done honestly). After the verifier computes the appropriate sums of commitments, the simulator sequentially executes several copies of the simulator of ``knowledge of a signature'' from Camenisch and Lysyanskaya (\cite{CL02} or \cite{CL04}, depending on the signature scheme used). After the remaining computations on the commitments by the verifier, the simulator finally runs the simulator for the proof of equality of committed values and the two range proofs.
\end{proof}
\end{document}